\documentclass[a4paper,12pt,reqno]{amsart}
\usepackage{amssymb}
\usepackage{enumerate}
\usepackage{ifthen}
\usepackage{graphicx}
\nonstopmode\numberwithin{equation}{section}
\newtheorem{proposition}{Proposition}[section]

\newtheorem{theorem}{Theorem}[section]

\newtheorem{lemma}{Lemma}[section]

\newtheorem{remark}{Remark}[section]

\allowdisplaybreaks

\begin{document}
\title{Bounded analytic maps, Wall fractions  and $ABC$--flow}

\author{
Alexei Tsygvintsev
}
\address{
U.M.P.A, Ecole Normale Sup\'{e}rieure de Lyon\\
46, all\'{e}e d'Italie, F69364 Lyon Cedex 07
}
\email{alexei.tsygvintsev@ens-lyon.fr}

\bigskip
\begin{abstract}
In this work we study the  qualitative properties of real analytic bounded maps defined in the infinite complex strip.  The main tool is approximation  by   continued $g$--fractions of Wall \cite{W}.  As an application, the $ABC$--flow system is considered which is essential to the origin of the solar magnetic field \cite{David}. \end{abstract}.  


\maketitle
\pagestyle{myheadings}
\markboth{ 
Alexei Tsygvintsev
}{
Continued fractions,  real analytic maps,  dynamical systems, ABC flow}

\section{Introduction}

In 1948 Hubert Wall introduced  the particular class of functional  continued fractions called  $g$--fractions. 
The  objective of the present study is to broaden our understanding of Wall's ideas in the dynamical system theory. 

In this section we will remind the reader of some key  facts from the analytic theory of continued fractions.    Let  $\mathbb H=\mathbb C_-\cup \mathbb C_+ \cup (-1,+\infty)$ where $\mathbb C_+=\{ z\in \mathbb C\, : \, \mathrm{Im}(z)>0\}$,  $\mathbb C_-=\{ z\in \mathbb C\, : \, \mathrm{Im}(z)<0\}$.

 For arbitrary  real sequence $g_i\in [0,1 ]$, $i\geq 1$ we call $g$--{\it fraction}  the continued  fraction 
\begin{align} \label{f4}
 g(z)=\{ g_1,g_2,... | z  \}=\dfrac{1}{1}
\begin{array}{cc}\\+\end{array}
\dfrac{g_1z}{1}
\begin{array}{cc}\\+\end{array}
\dfrac{(1-g_1)g_2z}{1}
\begin{array}{cc}\\+\end{array}
\dfrac{(1-g_2)g_3z}{1}
\cdots\,,
\end{align}
  converging  uniformly on compact sets  of $\mathbb H$ to an analytic function (see \cite{W} ).  The map $g$  is   rational   if and only if $g_k \in \{0,1\}$, for some $k\geq1$.  

In particular, if  $g_i=p\in (0,1)$, $i\geq 1$ then $g$ is algebraic and is given explicitly by 
\begin{equation}
\{p,p,\dots \mid z \}=\frac{ 2(1-p) }{ 1-2p+\sqrt{1+4p(1-p)z}   }, \quad z\in \mathbb C_-\cup \mathbb C_+ \cup (-{1}/{4p(1-p)},+\infty)\,.
\end{equation}
It is known that some ratios of hypergeometric functions can be expressed with help of $g$--fractions. Let $a,b,c$ are real constants satisfying $-1\leq a \leq c$, $0\leq b \leq c\neq 0$ and $F(a,b,c,z)$ be the hypergeometric function of Gauss. Then, as shown in  \cite{Kus}: 

\begin{equation}
\frac{F(a+1,b,c,-z)}{F(a,b,c,-z)}=\{ g_1,g_2,... | z  \}, \quad z \in \mathbb H\,,
\end{equation}
where
\begin{equation}
g_{2k}=\frac{a+k}{c+2k-1}, \quad   g_{2k-1}=\frac{b+k-1}{c+2k-2}, \quad k\geq 1\,.
\end{equation}
In the simplest case:
\begin{equation}
\frac{1}{z}\ln(1+z)=\frac{F(1,1,2,-z)}{F(0,1,2,-z)}\,.
\end{equation}

We define the truncated  continued $g$--fraction as  the $n$--order approximation of  \eqref{f4}:
 \begin{align} \label{truncated1}
\{g_1,g_2,...,g_n| z\}=\dfrac{1}{1}
\begin{array}{cc}\\+\end{array}
\dfrac{g_1z}{1}
\begin{array}{cc}\\+\end{array}
\dfrac{(1-g_1)g_2z}{1}
\begin{array}{cc}\\   \cdots \end{array}
\dfrac{(1-g_{n-1})g_nz}{1}\,,
\end{align}
which is a rational function of $z$ analytic in $\mathbb H$.

The next result is due to Gragg \cite{Gragg}: 

\begin{theorem}{(Gragg \cite{Gragg})}
Let $z\in(-1,+\infty)$, then the $n$--order truncation error satisfies:
\begin{equation}
|    \{ g_1,g_2,... | z  \}-   \{ g_1,g_2,...,g_n | z  \}  |\leq \left | 1- \frac{1}{1+z}\right | \left |  \frac{1-\sqrt{1+z  }}{ 1+\sqrt{1+z  } } \right |^n\,,
\end{equation}
and does not depend on values of $g_i$, $i\geq 1$.
\end{theorem}

The real  {\it a priori} bounds for the $g$--fraction are given by the next result:
\begin{theorem} 
\emph{(\cite{T1})}
\label{T1}\\
\noindent a) Let $k=2n+1$, $n=0,1,...$,  then
\begin{equation}
A_k(z)\leq g(z)\leq B_k(z),\quad -1<z<+\infty \,,
\end{equation}
where
\begin{equation}
A_k(z)=\{g_1,g_2,...,g_k| z\}, \quad  B_k(z)=\{g_1,g_2,...,g_k,1 | z\}\,.
\end{equation}

\noindent b)  Let $k=2n$, $n=1,2,...$, then
\begin{equation}
A_k^+(z)\leq g(z)\leq B_k^+(z),\quad 0\leq z<+\infty,
\end{equation}
\begin{equation}
A_k^-(z)\leq g(z)\leq B_k^-(z),\quad -1<z<0,
\end{equation}
where 
\begin{equation}
A_k^+(z)=\{g_1,g_2,...,g_k,1 | z\}, \quad B_k^+=\{g_1,g_2,...,g_k| z\} \,,
\end{equation}
and $A_k^-=B_k^+$, $B_k^-=A_k ^+$.
\end{theorem}

Using the above formulas we write below  the rational {\it a priori} bounds for the $g$--fraction  \eqref{f4}  corresponding to  $k=1,2,3$:

\noindent Case $k=1$.

\begin{equation}  \label{rrr}
A_1(z)=\frac{1}{1+g_1z},\quad B_1(z)=\frac{1+(1-g_1)z}{1+z}\,.
\end{equation}

\noindent Case $k=2$.

\begin{equation} \label{bounds2}
A_2^+(z)=\frac{(1-g_1g_2)z+1}{(1+z)(g_1(1-g_2)z+1)}, \quad B_2^+(z)=\frac{g_2(1-g_1)z+1}{(g_1-g_1g_2+g_2)z+1}\,,
\end{equation}
\begin{equation}  \label{bounds22}
 A_2^-=B_2^+, B_2^-=A_2^+\,.
\end{equation}

\noindent Case $k=3$.

\begin{equation}
A_3(z)=\frac{(g_3+g_2-g_3g_2-g_2g_1)z+1}{g_1g_3(1-g_2)z^2+(g_3+g_2+g_1-g_3g_2-g_1g_2)z+1}\,.
\end{equation}

\begin{equation}
B_3(z)=\frac{g_2(1-g_3)(1-g_1)z ^2+(1+g_2-g_3g_2-g_1g_2)z+1}{(1+z)((g_1+g_2-g_3g_2-g_1g_2)z+1)}\,.
\end{equation}

\begin{figure}[h] 
\includegraphics[scale=0.500,angle=0]{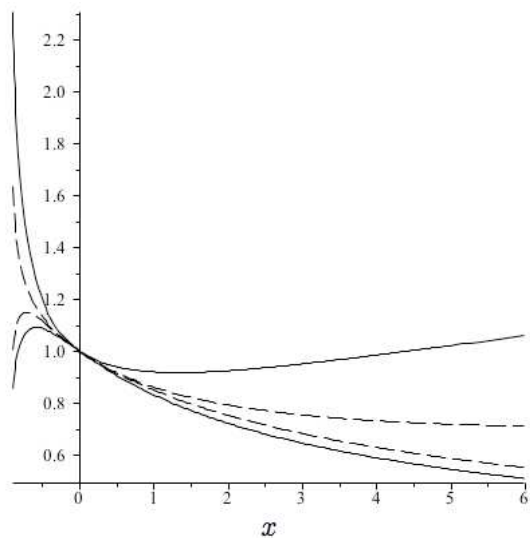} \caption{\label{fig2} Bounds  $ r(x) A_1(x) $,  $r(x)B_1(x)$ (bold line)   and  $r(x)A_2^{\pm}(x)$, $r(x)B_2^{\pm}(x)$ (dashed  line) for $x\in (-0.9,6)$,  $g_1=0.7$, $g_2=0.3$, $r(x)=\sqrt{1+x}$. }
\end{figure}

Let us put $g_1=0.7$, $g_2=0.3$ so that the rational bounds of order $1$ and $2$ given respectively by    ($A_1$, $B_1$) and ( $A^{\pm}_2$, $B^{\pm}_2$)  are defined.   The Figure \ref{fig2} illustrates  then the mutual position of  graphs of these functions, multiplied by $\sqrt{1+x}$ for $x \in (-0.9,6)$. As follows from Theorem \ref{T1}, the approximation of order $2$ is more precise  than the one given by order $1$.

The interesting link between $g$--fractions and probability theory was reported by Gerl  \cite{Gerl}.  One considers the nearest-neighbour random walks $X_n$, $n=0,1,2,...$ on $\mathbb N_0=\{0,1,2,...\}$ with the one-step transition  probabilities $p_{i,k}=\mathrm{Prob}[X_{n+1}=k\mid X_n=i ]$ defined by
\begin{equation}
p_{0, 1}=1,\quad p_{j,j-1}=g_j,\quad p_{j,j+1}=1-g_j\,
\end{equation}
with $0<g_j<1$, $j\geq 1$ and $p_{j,k}=0$ in any other case.

We introduce $p_{0,0}^{2n}=\mathrm{Prob}[X_{2n}=0 \mid X_0=0]$-- the probability of return to $0$ in $2n$ steps. The generating function for this sequence
\begin{equation}
G_0(z)=\sum _{n=0}^{\infty}\, (-1)^n\, p_{0,0}^{2n}\, z^n,\, 
\end{equation}
can be  written then   as a $g$--fraction:
\begin{equation}
G_0(z)=  \{ g_1,g_2,... | z  \}\,.
\end{equation}
Some other applications of $g$-fractions can be found in \cite{T1,T2,T3,T4}.

In the next section we will examine the relation between $g$--fractions and real analytic bounded functions.

\section{Functions bounded in the complex strip}

We denote   by $\mathbb A_{M,B}$   the set  of functions  $f(z)$ satisfying the following conditions

\noindent a) $f$  is holomorphic in the infinite strip $S_B=\{ z\in \mathbb C\, : | \mathrm{Im}(z) | <B \}$ , $B>0$.

\noindent b)  $f(\mathbb R)\subset \mathbb R$.

\noindent c)  $| f(z) | < M$, $M>0$,   $\forall $  $z\in S_B$.

\begin{figure}[h]
\includegraphics[scale=0.600,angle=0]{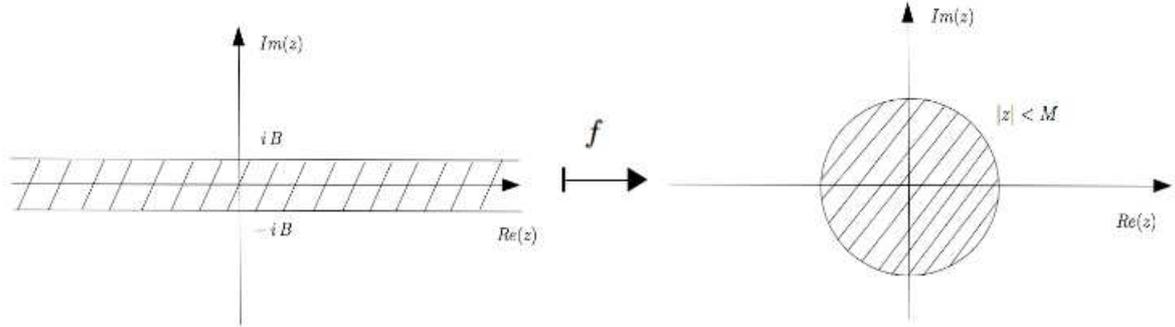} \caption{\label{fig4}  Map defined by  $f\in \mathbb A_{M,B}.$}
\end{figure}

Our   goal is to describe   the   $g$--fraction representation for the   elements of   $\mathbb A_{M,B}$ .
Firstly, we note that  it is sufficient to characterize only the functions from the class $\mathbb A_{1,\pi}$,  since   $f\in \mathbb A_{M,B}$ if and only if  $\theta(z)=\frac{1}{M}{f(Bz}/{\pi})\in  \mathbb A_{1,\pi}$.
The answer is given by the following theorem

\begin{theorem} \label{mainTh}
Let $\theta$ be a function holomorphic in the infinite complex strip  $S=\{ z\in \mathbb C\, : | \mathrm{Im}(z) | <\pi \}$, $|\theta(z)|<1$, $\forall \, z\in S$ and $\theta(\mathbb R)\subset \mathbb R$.  Then there exists a $g$-continued fraction 
\begin{align}  \label{maing}
 g(z)=\dfrac{1}{1}
\begin{array}{cc}\\+\end{array}
\dfrac{g_1z}{1}
\begin{array}{cc}\\+\end{array}
\dfrac{(1-g_1)g_2z}{1}
\begin{array}{cc}\\+\end{array}
\dfrac{(1-g_2)g_3z}{1}
\cdots\,,
\end{align}
with $g_k\in [0,1]$, $\forall \, k\geq 1$  such that
\begin{equation} \label{maingstrip}
 \theta(z)=1- \frac{2}{\mu\, e^{z/2}g(e^z-1)+1}, \quad z\in S\,,
\end{equation}
where  $\mu=\displaystyle\frac{1+\theta(0)}{1-\theta(0)}$.

\end{theorem}

\begin{proof}

We introduce the complex domains
\begin{equation} \label{domains}
\mathbb D=\{ z\in \mathbb C \,: | z|<1  \}, \quad \mathbb H_+=\{ z\in \mathbb C \,: \mathrm{Re}(z)>0  \}\,,
\end{equation}
and   the conformal maps defined by:   
\begin{equation} \label{mapm}
m(z)=\frac{1+z}{1-z},  \quad  m \,: \mathbb D \to \mathbb H_+\,,
\end{equation} 
and 
\begin{equation} \label{hhh}
 l(z)=\log(1+z), \quad  l\,:   \mathbb H \to S\,,
\end{equation}
\begin{equation} \label{rtr}
\eta(z)=l^{-1}(z)=e^z-1\,.
\end{equation}
We note that $\eta$   is a bijection between $\mathbb R$ and $(-1,+\infty)$.

One verifies that the composition $F=m \circ f \circ l $ is holomorphic in $\mathbb H$ and $F(\mathbb H)\subset \mathbb H_+$
with $F(z)\in \mathbb  R$ for $z>-1$.
Thus, according to theorem of Wall \cite{W}, p. 279 $F$ can be written as follows
\begin{equation}
F(z)=\mu \sqrt{1+z} \int_0^1 \displaystyle \frac{d \mu(u)}{1+zu}\,,
\end{equation}
for some nondecreasing real bounded function $\mu(u)$, $u\in (0,1)$ and $\mu>0$.

For $f=m^{-1}\circ F\circ \eta$ one obtains the following formula
\begin{equation} \label{g12}
f(z)= \left(  1-\displaystyle \frac{2}{\mu e^{z/2} \int_0^1\displaystyle\frac{d\mu(u) }{1+
( e^z-1)u}+1}     \right).
\end{equation}
The integral in \eqref{g12} can be transformed to the continued  $g$--fraction form  \cite{W}
\begin{equation}
   \int^1_0 \displaystyle \frac{d\mu (u)}{1+ (e^z-1)u}       
   =  \left  \{g_1,g_2,... | e^z-1 \right\}, \quad \mathrm{for \, \, some} \quad g_k\in[0,1]\,,
\end{equation}
that together with \eqref{g12} implies  \eqref{maingstrip}. End of proof.
\end{proof}

To calculate the coefficients $g_p$ in \eqref{maing}  one has  formulas:
\begin{equation}
g_p=C_p(\theta(0),\theta'(0),..., \theta^{(p)}(0) ), \quad p\geq 1\,,
\end{equation}
with rational functions $C_p$  determined by calculation of derivatives  of both sides of \eqref{maingstrip} at $z=0$.   
The  recurrent formulas for  all $C_p$ can be derived from \cite{W}, p. 203.

Introducing
\begin{equation} \label{thetan}
 \theta_n=\theta^{(n)}(0), \quad n\geq 0\,,
\end{equation}
 we provide below explicit formulas for $g_1$ and  $g_2$: 

\begin{equation} \label{g1strip}
g_1= \frac{1}{2} \, \frac{ 1-4\theta_1-\theta_0^2} { 1-\theta_0^2  }   \,,
\end{equation}

\begin{equation} \label{g2strip}
g_2=  \frac{1}{2}\frac{(16\theta_1^2-8\theta_2-\theta_0+\theta_0^2+\theta_0^3-8\theta_2\theta_0-1)(1-\theta_0)} {(1-\theta_0^2+4\theta_1)(4\theta_1-1+\theta_0^2)}  \,.
\end{equation}

Our present aim is  to estimate  the time of  return of $\theta \in \mathbb A_{1,\pi}$   to the initial value $\theta(0)$ i.e to study the real points $\tau  \neq 0$   such that $\theta(\tau)=\theta(0)$.   For this sake,   we will  use the  {\it a priori} bounds  \eqref{rrr} applied to the $g$--fraction in formula  \eqref{maingstrip}. For $p=2k+1$ one obtains:

\begin{equation} \label{main1}
\left(      1- \displaystyle \frac{2}{\mu  \sqrt{1+\eta(z)}  A_p(\eta(z)) +1 }       \right)  \leq \theta(z) \leq \left(      1- \displaystyle \frac{2}{\mu  \sqrt{1+\eta(z)}  B_p(\eta(z)) +1 }       \right)\,,
\end{equation}
for $z\in \mathbb R$. 

If $p=2k$ then 
\begin{equation} \label{main2}
\left(      1- \displaystyle \frac{2}{\mu  \sqrt{1+\eta(z)}  A_p^+(\eta(z)) +1 }       \right)  \leq   \theta(z) \leq \left(      1- \displaystyle \frac{2}{\mu  \sqrt{1+\eta(z)}  B_p^+(\eta(z)) +1 }       \right)\,,
\end{equation}
for $z\in (0,+\infty)$, and

\begin{equation} \label{main21}
\left(      1- \displaystyle \frac{2}{\mu  \sqrt{1+\eta(z)}  A_p^-(\eta(z)) +1 }       \right)  \leq   \theta(z) \leq \left(      1- \displaystyle \frac{2}{\mu  \sqrt{1+\eta(z)}  B_p^-(\eta(z)) +1 }       \right)\,,  
\end{equation}
for $z\in (-\infty,0]$.

In the next theorem,  for a given $f\in \mathbb A_{1,\pi}$, we will  describe a neighborhood of origin in which $x=0$ is the only solution of $\theta(x)=\theta(0)$.

\begin{theorem} \label{theorem1}
We assume that  all conditions of the Theorem \ref{mainTh} are fulfilled and there exists a non zero $\tau \in \mathbb R$ such that $\theta(\tau)=\theta(0)$. Then $g_1(1-g_1)\neq 0$ and the following inequality  holds
\begin{equation}  \label{bound1} 
|\tau| \geq 2 \left | \log \frac{1-g_1}{g_1} \right | \,.
\end{equation}

 \end{theorem}

\begin{proof}
One considers   \eqref{main1} with $p=1$.  We have  $A_1(0)=B_1(0)=1$  and define  $t_{1}$,  $t_2$  as non-zero  solutions  of the following algebraic equations
\begin{equation} \label{equations}
\sqrt{1+t_1}A_1(t_1)=1, \quad \sqrt{1+t_2}B_1(t_2)=1, \quad t_1,t_2\in (-1,+\infty)\,.
\end{equation}

Simple algebraic calculations show that the only  solutions satisfying \eqref{equations}  are  given by
\begin{equation} 
  t_1=\frac{1-2g_1}{g_1^2}, \quad t_2=\frac{2g_1-1}{(1-g_1)^2}\,,
\end{equation}
which are related by
\begin{equation}
\frac{1}{t_1}+\frac{1}{t_2}=-1\,.
\end{equation}
Since $\eta(z)$ is a bijection between $\mathbb R$  and $ (-1,+\infty)$ there exist unique real numbers  $T_1$, $T_2\in \mathbb R$ satisfying the following equations:
\begin{equation}
 \eta(T_{1})=t_{1}, \quad \eta(T_{2})=t_{2}\,.
\end{equation}
As easily seen from  \eqref{rtr}: $T_1=-T_2$ and the proof of \eqref{bound1} follows straightforwardly  from the formula \eqref{hhh}.
\end{proof}

\begin{remark} It is easy to see that  $g_1=1/2$ if and only if $\theta'(0)=0$, as follows from    \eqref{g1strip}.   Hence, the right hand side of the inequality  \eqref{bound1}  is zero if 
$\theta'(0)=0$, that corresponds to trivial case and is strictly positive if $\theta'(0)\neq 0$.
\end{remark}

The next result shows that   $\theta \in  \mathbb A_{1,\pi}$, under some conditions on derivatives $\theta^{(p)}(0)$, $p=0,1,2$,  always returns to the initial value $\theta(0)$  i.e admits the oscillatory property.

\begin{theorem} \label{2strip}
Let $\theta \in \mathbb A_{1,\pi}$, $\theta'(0)>0$  and   $g_1,g_2$ are defined by formulas  \eqref{thetan}  and   \eqref{g1strip}, \eqref{g2strip}.

\noindent I. \,  We assume that  the $(g_1,g_2)\in [0,1]^2$  belongs to one of the two  regions $E$, $F$  defined by:   
\begin{align}
E=\{ (g_1,g_2)\in (0,1)^2 \,: \, D_2 \geq 0,  0< g_1<1/2,0<g_2<1/2 \}\,, \label{k1} \\
F=\{ (g_1,g_2)\in (0,1)^2 \,: \, D_2 \geq 0,  0< g_1<1/2,  1/2<g_2<1   \}\,, \label{k2} \\
D_1=(1-g_1)^2-4g_1^2g_2(1-g_2), \, D_2=g_1^2-4(1-g_1)^2(1-g_2)g_2\,.
\end{align}
Let
\begin{align}
\zeta= \frac{2B}{\pi} \log(t^{(2)}_1) >0 \quad  \mathrm{if}  \quad  (g_1,g_2)\in E \label{imp} \\
\zeta= \frac{2B}{\pi} \log(t^{(2)}_2)<0  \quad  \mathrm{if}  \quad  (g_1,g_2)\in F 
\end{align}
where $t^{(i)}_j$, $i,j=1,2$ are defined as functions of $g_1,g_2$ by 
\begin{equation} \label{t1}
t_{2,1}^{(1)}=\frac{1-g_1\pm\sqrt{D_1}}{2g_1(1-g_2)},\quad 
t_{2,1}^{(2)}=\frac{g_1\pm\sqrt{D_2}}{2(1-g_1)g_2}\,.
\end{equation}

Then there exists $\tau \in \mathbb R$, $\tau \neq 0$  such that 
\begin{equation}  \label{INN1}
 \theta(\tau)=\theta(0)\,,
\end{equation}
and  
\begin{equation} \label{INCLUSION}
\tau\in (0,\zeta)\,  \quad \mathrm{if}\quad  \zeta>0 \quad \mathrm{ and } \quad  \tau\in (\zeta,0) \quad  \mathrm{if}  \quad   \zeta<0\,.   
\end{equation}

\noindent II. \, Let  $\tau \in \mathbb R$, $\tau \neq 0$ be  such that $\theta(\tau)=\theta(0)$, then

\begin{equation} \label{55}
\tau \in (-\infty, \frac{2B}{\pi} \log(t_1^{(1)})]\cup[ \frac{2B}{\pi} \log(t_2^{(1)})  , +\infty)\,. \end{equation}
\end{theorem}

\begin{figure}[h]
\includegraphics[scale=0.500,angle=0]{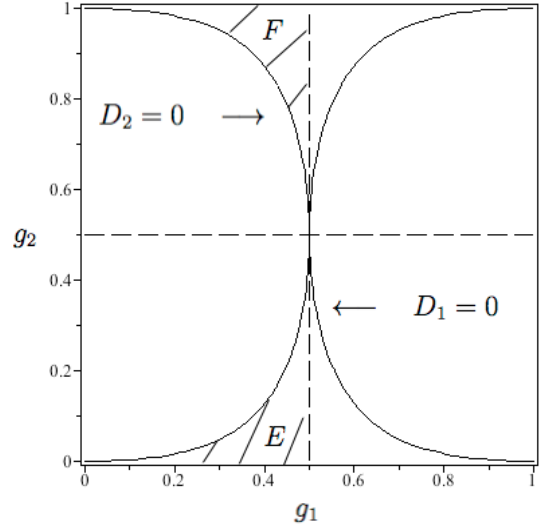} \caption{\label{fig1} Domains  $E$, $F$  in the parameter  space $(g_1,g_2)\in [0,1]^2$.}
\end{figure}

\begin{proof}
We  consider  \eqref{main2}  with   $p=2$  and define  the following real algebraic equations
\begin{align*}
 \sqrt{1+x}\,A_2^+(x)=1\,, \quad x \in (0,+\infty)\,, \quad \mathrm{(A_1)}\,,\\
 \sqrt{1+x}\,B_2^+(x)=1\,, \quad x \in (0,+\infty)\,,  \quad \mathrm{(B_1)}\,, \\
 \sqrt{1+x}\,B_2^-(x)=1\,, \quad x \in (-1,0)\,,  \quad \mathrm{(\tilde A_1)}\,, \\
 \sqrt{1+x}\,A_2^-(x)=1\,, \quad x \in (-1,0)\,.\quad \mathrm{(\tilde B_1)}\,,
\end{align*}
where $A_2^-=B_2^+$, $B_2^-=A_2^+$.

Making  the change of variables 
\begin{equation} \label{changef}
x=-1+t^2, \quad t\in \mathbb R\,,
\end{equation}
after some elementary  transformations, it is easy to show that  equations  ($A_1$), ($B_1$) are  equivalent respectively to quadratic equations ($A_2$) and    ($B_2$) given below
\begin{align*}
  P_1(t)=g_1(1-g_2)t^2-(1-g_1)t+g_1g_2=0,   \quad t\in \mathbb R,    \quad \mathrm{(A_2)}\\
P_2(t)=g_2(1-g_1)t^2-g_1t+(1-g_1)(1-g_2)=0,\quad t\in \mathbb R\,. \quad \mathrm{(B_2)}
\end{align*}
\begin{remark}
We notice that $P_2(t)$ is obtained  from $P_1(t)$  by transformation
\begin{equation}
g_i\mapsto 1-g_i, \quad  i=1,2\,.
\end{equation}
\end{remark}

The  polynomial $P_1(t)=0$ has  two real  roots   $t_1^{(1)},t_2^{(1)}\in \mathbb R$
\begin{equation} 
t_1^{(1)}=\frac{1-g_1-\sqrt{D_1}}{2g_1(1-g_2)}, \quad t_2^{(1)}=\frac{1-g_1+\sqrt{D_1}}{2g_1(1-g_2)}, \quad t_1^{(1)}\leq t_2^{(1)}\,,
\end{equation}
if and only if the following condition holds 
\begin{equation} \label{D1}
D_1=(1-g_1)^2-4g_1^2g_2(1-g_2)\geq 0\,.
\end{equation}
$P_2(t)=0$ has  two real  solutions   $t_1^{(2)},t_2^{(2)}\in \mathbb R$
\begin{equation} \label{t2}
t_1^{(2)}=\frac{g_1-\sqrt{D_2}}{2(1-g_1)g_2}, \quad t_2^{(2)}=\frac{g_1+\sqrt{D_2}}{2(1-g_1)g_2}, \quad t_1^{(2)}\leq t_2^{(2)}\,,
\end{equation}
if and only if 
 \begin{equation} \label{D2}
D_2=g_1^2-4(1-g_1)^2(1-g_2)g_2\geq  0\,.
\end{equation}
Applying the  Vieta's formulas to polynomials $A_2$ and $B_2$, and taking into account that $g_i\in (0,1)$, $i=1,2$  one checks that:
\begin{equation} \label{VI}
t^{(i)}_j> 0, \quad  i,j=1,2\,.
\end{equation}

\noindent Let $f'(0)>0 ( \Leftrightarrow g_1<1/2$). Then $\theta(z)$ is increasing function in the interval $(-\epsilon,\epsilon)$ for some small $\epsilon>0$. We assume that inequality  $D_2 \geq 0$ holds, so both roots $t^{(2)}_1$ and $t^{(2)}_2$ are real. One has $P_2(1)=1-2g_1>0$, so, in view of \eqref{VI},  either 
$$0<t^{(2)}_1\leq t^{(2)}_2<1 \eqno{(a)}$$ 
or 
$$1<t^{(2)}_1\leq t^{(2)}_2 \eqno{(b)}$$

One verifies with help of  \eqref{t2} that (a) is equivalent to 
$$L_2=g_1-2(1-g_1)g_2<0$$
  and (b) to  $L_2>0$.

Thus,  in view of \eqref{changef},  if (b) holds, the equation ($B_1$) will have  solution   
$$ x=-1+{t^{(2)}_1}^2 \in (0, +\infty)\,,$$
and if (a) holds,  ($\tilde B_1$) will have solution  
 $$x=-1+{t^{(2)}_2}^2 \in (-1, 0)\,,$$
 in view of \eqref{changef}.

\noindent One verifies directly that the condition $0<g_1<1/2$ implies $D_1>0$.

So, the   both roots $t_1^{(1)}$ and $t_2^{(1)}$ are real distinct numbers.

Since $P_1(1)=2g_1-1<0$ we have $0<t_1^{(1)}<1<t_2^{(1)}$.  So,  the equation ($A_1$) will have the unique real  solution 
\begin{equation}
y_2=-1+{t_2^{(1)}}^2\in (0, +\infty)\,,
\end{equation}
and  ($\tilde A_1$) will have the unique real  solution
\begin{equation}
 y_1=-1+{t_1^{(1)}}^2\in (-1,0)\,.
\end{equation}

Since $\eta(z)$ is a bijection of $\mathbb R$ and  $(-1,+\infty)$, there exists unique real number  $\zeta\in \mathbb R$
 satisfying equation $\eta(\zeta)=x$ with $x\in(-1,+\infty)$ defined  above. Then, as follows from  \eqref{main2}, \eqref{main21},  there exists  $\tau$ satisfying  \eqref{INN1}  if one of the cases \eqref{k1}-\eqref{k2} holds.  One has   $\tau\in (0,\zeta)$ if $\zeta>0$ and  $\tau\in (\zeta,0)$ if  $\zeta<0$.   

Using $y_{1,2}$ defined above, we define $z_1<0$ and $z_2>0$ as unique real  solutions of $\eta (z_i)=y_i$, $i=1,2$.   Let now $\tau \neq 0$ be such that $\theta(\tau)=\theta(0)$, then  $ \tau \in (-\infty,z_1]\cup [z_2,+\infty)$  that shows \eqref{55}  and finishes the proof.

\end{proof}

 The case  $\theta'(0)<0$ can be analyzed in the similar  way by considering   $\theta(-z)$ instead of $\theta(z)$.

\section{Applications to solutions of  the $ABC$--flow equations} \label{Ap}
The $ABC$--flow  is a system of three ordinary differential equations
\begin{equation} \label{abc}
\frac{dx_1}{dz}=  A\sin x_3+C\cos x_2, \quad \frac{d x_2}{dz}=  B\sin x_1+A\cos x_3, \quad  \frac{d x_3}{dz}=  C\sin x_2+B\cos x_1\, , 
\end{equation}
depending on  three arbitrary  real positive  constants  $(A,B,C)\neq(0,0,0)$.

This vector field    appears as an exact solution of the Euler equation without forcing. It is essential to  the origin of magnetic fields in large astrophysical bodies like the Earth, the Sun and galaxies. The  history of the problem,  including numerous applications,  can be found  in   \cite{David}, \cite{Ash}. Through intensive numerical studies \cite{Dombre}, \cite{X1}, \cite{X2}, it was shown that the dynamics of the $ABC$-flow is generally chaotic.  In particular, it means that, due to exponential instability of its solutions, any kind of a long time prediction of evolutionary dynamics  is problematic.

The rigorous study of integrability, that is   of the existence of conservation lows (first integrals) of the $ABC$-flow  was initiated by Ziglin in the cycle of papers \cite{Z1}-\cite{Z3}.  Applying the complex analytic monodromy approach, he proved  that the $ABC$-flow does not have any non constant meromorphic first integrals. The results of Ziglin have been  generalized and extended later in \cite{Mac} based on the more recent  differential Galois approach of Morales-Ramis  (see \cite{Morales} for details).

The purpose of our  study is  to elaborate on the idea that  some qualitative insight into dynamics of the $ABC$--flow   can be   gained  using the $g$-fractions approach.  More precisely, we aim to establish analytically  the existence of   the  recurrent behavior (Theorem \ref{TTTT}).

We define
\begin{equation} \label{delta}
\delta=\max \{ A+C,B+A,C+B \}>0\,.
\end{equation}

Since the vector field of \eqref{abc} is a bounded one it is complete in $\mathbb R^3$ and hence all its real solutions $(x_1(z),x_2(z),x_3(z))$ are defined for $z\in \mathbb R$.  We will need the following   elementary property, easy to proof
\begin{proposition} \label{boundsTrig}
The trigonometric functions $\cos x$, $\sin x$ are analytic and  bounded in absolute value by $\mathrm{ch}( \epsilon)$ in the complex disk $|x-x_0|<\epsilon$ with a center $x_0\in \mathbb R$.
\end{proposition}

And we remind   the  classical theorem of Picard (see for example \cite{Picard}) from  the analytic theory of ordinary differential equations:

\begin{theorem}{(Picard \cite{Picard})} \label{picpic}
Let  $Q_j(q_1,q_2,\dots,q_n)$, $ 1\leq j \leq n$ be analytic functions in the complex domain
\begin{equation} \label{cond}
|q_i-\bar q_i|<q'_i,  \quad 1\leq i \leq n\, ,
\end{equation}
for some $q'_i>0$, $\bar q_i\in \mathbb C$,  $1\leq i \leq n$.
We assume that there exist positive constants $Q'_j>0$, $1\leq j \leq n$ such that 
\begin{equation}
|Q_j(q_1,q_2,\dots,q_n)|<Q'_j, \quad 1\leq j\leq n\,,
\end{equation}
if the conditions \eqref{cond} hold.

Then the system of $n$ ordinary equations
\begin{equation}
\frac{dq_j}{d z}=Q_j(q_1,q_2,\dots,q_n),\quad 1\leq  j \leq n\,,
\end{equation}
admits the  unique  solution analytic in the complex disc
\begin{equation}
D\, :\,|z|<T',  \quad T'=\mathrm{min} \left\{  \frac{ q'_1}{Q'_1},\frac{ q'_2}{Q'_2},\dots,\frac{ q'_n}{Q'_n}       \right \}\,,
\end{equation} 
satisfying the  initial conditions $q_i(0)=\bar q_i$, $1\leq i \leq n$.

Moreover, 
\begin{equation} \label{PE}
|q_i(z)-\bar q_i|<q'_i, \quad  z\in D,\quad 1\leq i \leq n\,.
\end{equation}

\end{theorem}

Our first result is given  by the following theorem

\begin{lemma} Let $\epsilon>0$ be an arbitrary positive real number. Then every real solution  $(x_1(z),x_2(z),x_3(z))$ of \eqref{abc} is an  analytic  function of $z$  in the complex infinite strip    
\begin{equation} \label{strip}
S_b=     \{ z\in \mathbb C \, : \,  |\mathrm{Im}( z)|<b \}, \quad  b=\frac{1}{\delta}\, \frac{\epsilon  }{ \mathrm{ch} (\epsilon)}\,.
\end{equation}

 \end{lemma}

\begin{proof}
 For a given $\epsilon>0$ and a real triplet $(X_1,X_2,X_3)\in \mathbb  R^3$,  we introduce the complex domain $\mathcal D_{X_1,X_2,X_3}^{\epsilon}=\{ (x_1,x_2,x_3)\in \mathbb C^3  \, :\,  | x_i-X_i |<\epsilon, \, i=1,2,3 \}$.   Writing the  system \eqref{abc} in the form $\dot x_i=f_i(x_1,x_2,x_3)$, $i=1,2,3$ one easily verifies  that 
\begin{equation}
 |f_i(P)|<\delta \, \mathrm{ch} (\epsilon), \quad \forall \,  P=(x_1,x_2,x_3) \in { \mathcal D_{X_1,X_2,X_3}^{\epsilon}}, \quad  1\leq i \leq 3\,, 
\end{equation}
in view of  Proposition  \ref{boundsTrig}.

According to the Picard's Theorem \ref{picpic}, the   solution of \eqref{abc}  $z\mapsto (x_1(z),x_2(z),x_3(z))$,   defined by   the initial condition $(x_1(z_0),x_2(z_0),x_3(z_0))=(X_1,X_2,X_3)$,    $z_0\in \mathbb R$,  will be analytic in the disk $D_{z_0,T'}\,:\, |z-z_0|<T'$ of the complex time plane with  $T'=\displaystyle \frac{\epsilon  }{ \delta \, \mathrm{ch} (\epsilon)}$.  Since the system \eqref{abc} is autonomous  and $S_b=\cup _{z_0\in \mathbb R}\, D_{z_0,T'}$ that finishes the proof. 

\end{proof}

\begin{remark}
The function $e(\epsilon)=\epsilon  /   \mathrm{ch}( \epsilon)$ reaches the unique maximal value for $\epsilon\in \mathbb R_+$ which we denote $e_{max}$. The direct computation gives $e_{max}=e(\epsilon_{max})=0.6627$  for $\epsilon_{max}=1.1997$.
\end{remark}

The solutions of the  $ABC$--flow are   defined on the  torus $\mathbb T^3=\{(x_1,x_2,x_3) \, :\,  x_i\, \mathrm{mod}\, 2\pi  \}$.   We aim now   to study  the projections  $\phi_i\, :\, \mathbb R  \to [-1,1]$ defined by 
\begin{equation}
\phi_i(z)=\sin( x_i(z)/2), \quad i=1,2,3\,.
\end{equation}
These functions  are of dynamical importance since  the sections  $x_i=0\mod  2\pi$   in $\mathbb T^3$ can be  viewed  as the zero level surfaces $\phi_i=0$. Indeed,  since $x_i(z)\in [0,2\pi]$ the only solutions of the equation $\phi_i(z)=\sin( x_i(z)/2)=0$  is given by  $x_i(z)=0 \mod  2\pi$.

\begin{lemma} \label{LEM}
Let $z\in \mathbb R \mapsto (x_1(z),x_2(z),x_3(z))$ be an arbitrary  real solution of the $ABC$--flow  \eqref{abc}.  Then all functions $\phi_i(z)=\sin( x_i(z) /2) $, $i=1,2,3$ are analytic in $S_b$ defined by \eqref{strip}  and bounded in absolute value by $M=  \mathrm{ch} (\epsilon/2)$.
\end{lemma}
\begin{proof}
Let $x_j(z)\in (x_1(z),x_2(z),x_3(z))$.  Then, $\phi_j(z)=\sin (x_j(z))$ is analytic in $S_b$ as a composition of analytic maps. Let $z_0\in \mathbb R$, then for  the complex  disc $D_{z_0,b}\, :\,  |z-z_0|<b$  we have obviously $D_{z_0,b}\subset S_b$.  According to \eqref{PE} $\forall \, z \in D_{z_0,b}\, :\, |x_j(z)-x_j(z_0)|<\epsilon$. That  completes the proof in view of  $S_b=\cup _{z_0\in \mathbb R}\, D_{z_0,b}$ and Proposition \ref{boundsTrig}.

\end{proof}

We consider  an arbitrary solution  $\Gamma_{\alpha,\beta}$ of \eqref{abc} starting from the plane $x_1=0 \mod 2 \pi$, and defined by initial conditions of the form 
\begin{equation} \label{temp}
x_1(0)=0, \quad  x_2(0)=\alpha, \quad  x_3(0)=\beta,  \quad \alpha,\beta \in [0,2 \pi]\,.
\end{equation}
 Let $f(z)=\phi_1(z)=\sin(x_1(z)/2)$.  We have $f(0)=0$, $f \in \mathbb A_{M,b}$ and the formulas \eqref{thetan}, \eqref{g1strip}, \eqref{g2strip} give
\begin{equation} \label{suite}
\theta_0=0,\quad \theta_1=\frac{1}{2\pi \delta}\, \frac{\epsilon}{\mathrm{ch}(\epsilon)\mathrm{ch}(\epsilon/2)} (A\sin \beta+C \cos \alpha), \quad  \theta_2=\frac{1}{2\pi^2\delta^2} \, \frac{\epsilon^2}{\mathrm{ch}^2\mathrm{ch}(\epsilon)(\epsilon/2)} AB\cos \beta\,,
\end{equation}
and
\begin{equation}
g_1=\frac{1}{2} \left (  1-4\theta_1   \right), \quad g_2=\frac{1}{2}\left  (1+ \frac{8\theta_2}{1-16\theta_1^2}\right) \,.
\end{equation}

Since $f \in \mathbb A_{M,b}$, according to Theorem \ref{mainTh} we have $g_{1,2}\in [0,1]$. One easily verifies  that the strict inequalities always hold:     $0<g_{1,2}<1$.
Let $\tau\neq 0$ be such that $x_1(\tau)=0 \mod 2 \pi$.  According to Theorem \ref{theorem1} we have the following explicit  lower bound
\begin{equation} \label{5}
|  \tau |    \geq  \frac{4 \, \epsilon}{\pi \,\delta\,  \mathrm{ch}(\epsilon)} \mathrm{arctanh} \left(  \frac{2\,  \epsilon}{   \pi \, \delta\, \mathrm{ch}(\epsilon/2) \mathrm{ch}(\epsilon)   }     ( A\sin \beta+C\cos \alpha   ) \right)  \,,
\end{equation}
which is true for any $(\alpha,\beta)\in [0,2\pi]^2$ and  arbitrary $\epsilon>0$, $\delta$ is defined by \eqref{delta}. 

From the dynamical point of view, this  means that the solution $\Gamma_{\alpha,\beta}$ leaving  the plane $x_1=0$,  can not return back earlier than permitted   by \eqref{5}.  In practice,  one can use   a freedom in the  chose of   $\epsilon$ in order to make  \eqref{5} optimal.  The more precise lower bound, involving $g_2$ i.e the second derivative of $f$ at $0$, is given by \eqref{55}. The corresponding to \eqref{5} formula  is straightforward   to obtain.

Now we will analyze the upper bound on $\tau$  given by Theorem \ref{2strip} assuming that $\dot x_1(0)>0 \Leftrightarrow g_1<1/2$.

The condition  \eqref{k1}  defines  the region $E$ in the parameter space $(\alpha,\beta)\in [0,2\pi]^2$ and   is   equivalent to the system of inequalities 
\begin{equation} \label{it}
 \theta_1>0,\quad \theta_2\leq -\frac{\sqrt{\theta_1}}{2}(1-4\theta_1)\,.
\end{equation}
One can  show that for any positive constants $A,B,C$ and $\epsilon$  the above set  of parameter values  $(\alpha,\beta)$ is  non empty in  $[0,2 \pi]^2$  . Indeed,  it is sufficient to put $\alpha=\pi/2$  in  \eqref{suite}  and consider  the values ($\beta \to \pi$, $\beta <\pi$) to satisfy \eqref{it}.  Therefore, by the continuity argument, we can state  the following 

\begin{theorem} \label{TTTT}

The Poincar\'e section $\Pi =\{ (x_1,x_2,x_3)\in \mathbb T^3 \, | \, x_1=0 \mod 2 \pi\}$ of the $ABC$--flow \eqref{abc}  contains  a positive measure set of initial conditions  $(\alpha,\beta)$ for which the solutions $\Gamma_{\alpha,\beta}$ return to $\Pi$.   
 \end{theorem}
 
The corresponding upper bounds for the time of return  $\tau$ can be calculated with help of    \eqref{imp}. 

\section{Conclusion}

The presented work was motivated  by the ideas  of Poincar\'e and Sundman \cite{Diacu} from the Celestial Mechanics  providing converging time series solutions for the $3$--body problem.    These ones  follows  from the analyticity of regularized  collision free solutions   in the infinite complex strip of the time plane. Unfortunately, these series, though existing  for all values of time, have very slow convergence. One would have  to sum up milliards of terms  to gain any significant qualitative information  above the motion of particles.  The present study was designed to test the hypothesis that this gap  can be overcome by replacing the power series with functional continued fractions.   We notice that this  issue has not yet been addressed fully in the literature and a number of aspects of the $g$--fractions approach  presented here require further investigation.   As a toy model,  we consider the  $ABC$--flow  system \eqref{abc}   exhibiting the chaotic behavior.  In particular, this system can not be solved by quadratures and being non--integrable, does not have any  globally defined analytic  conservation  laws.   
Nevertheless,  as shown by Lemma \ref{LEM}, all solutions of this system belong to the class of real analytic bounded functions and thus admit the $g$--fraction representation as stated  by Theorem \ref{mainTh}.  Applying the rational approximation, we can derive the upper and lower bounds on the time of the first return  for a Poincar\'e map of the $ABC$-flow which was defined in Section  \ref{Ap}.  As a conclusion, we establish by  Theorem \ref{TTTT} the  existence of a  non empty set  of trajectories   of the $ABC$--flow,  exhibiting  the recurrent behavior. 

\section{Acknowledgments}
The author is grateful to  anonymous referees for their useful remarks and suggestions.

\end{document}